\documentclass[journal]{IEEEtran}

\usepackage{amsmath}
\usepackage{amsfonts}
\usepackage{amssymb}
\usepackage{amsthm}
\usepackage{tabularx}
\usepackage{mathrsfs} 

\usepackage{url}
\newtheorem{theorem}{Theorem}
\newtheorem{remark}{Remark}

\newtheorem{corollary}{Corollary}

\usepackage{stfloats}
\usepackage{float}
\usepackage{graphicx}
\usepackage{xcolor}
\usepackage{subfigure}
\usepackage[colorlinks=true, allcolors=blue]{hyperref}

\renewcommand{\vec}[1]{\mathbf{#1}}

\makeatletter
\def\blfootnote{\xdef\@thefnmark{}\@footnotetext}
\makeatother

\begin{document}
	
\title{Fluid Antenna Multiple Access with\\Simultaneous Non-unique Decoding in\\Strong Interference Channel} 

\author{Farshad~Rostami~Ghadi,~\IEEEmembership{Member},~\textit{IEEE}, 
             Kai-Kit~Wong,~\IEEEmembership{Fellow},~\textit{IEEE},
             Masoud~Kaveh,\\ H. Xu,~\IEEEmembership{Member},~\textit{IEEE}, 
             W.~K.~New,~\IEEEmembership{Member},~\textit{IEEE},
             F. Javier~L\'{o}pez-Mart\'{i}nez,~\IEEEmembership{Senior Member},~\textit{IEEE}\\ 
             and Hyundong Shin,~\IEEEmembership{Fellow, IEEE}
             \vspace{-3mm}
             }
\maketitle

\begin{abstract}
Fluid antenna system (FAS) is gaining attention as an innovative technology for boosting diversity and multiplexing gains. As a key innovation, it presents the possibility to overcome interference by position reconfigurability on one radio frequency (RF) chain, giving rise to the concept of fluid antenna multiple access (FAMA). While FAMA is originally designed to deal with interference mainly by position change and treat interference as noise, this is not rate optimal, especially when suffering from a strong interference channel (IC) where all positions have strong interference. To tackle this, this paper considers a two-user strong IC where FAMA is used in conjunction with simultaneous non-unique decoding (SND). Specifically, we analyze the key statistics for the signal-to-noise ratio (SNR) and interference-to-noise ratio (INR) for a canonical two-user IC setup, and subsequently derive the delay outage rate (DOR), outage probability (OP) and ergodic capacity (EC) of the FAMA-IC. Our numerical results illustrate huge benefits of FAMA with SND over traditional fixed-position antenna systems (TAS) with SND in the fading IC.
\end{abstract}


\maketitle
\blfootnote{The work of F. Rostami Ghadi, K. K. Wong, H. Xu, and W. K. New is supported by the Engineering and Physical Sciences Research Council (EPSRC) under Grant EP/W026813/1. The work of F. J. L\'{o}pez-Mart\'{i}nez is funded in part by Junta de Andalucia through grant EMERGIA20-00297, and in part by MICIU/AEI/10.13039/50110001103 through grant PID2023-149975OB-I00 (COSTUME).}
\blfootnote{\noindent F. Rostami Ghadi, K. K. Wong, H. Xu, and W. K. New are with the Department of Electronic and Electrical Engineering, University College London, London, UK. K. K. Wong is also affiliated with the Department of Electronic Engineering, Kyung Hee University, Yongin-si, Gyeonggi-do 17104, Korea. (e-mail: $\rm \{f.rostamighadi, kai\text{-}kit.wong, hao.xu,  a.new\}@ucl.ac.uk$).}
\blfootnote{\noindent M. Kaveh is  with the Department of Information and Communication Engineering, Aalto University, Espoo, Finland. (e-mail: $\rm masoud.kaveh@aalto.fi$).}
\blfootnote{\noindent F. J. L\'{o}pez-Mart\'{i}nez is with the Department of Signal Theory, Networking and Communications, Research Centre for Information and Communication Technologies (CITIC-UGR), University of Granada, 18071, Granada, Spain. (e-mail: $\rm fjlm@ugr.es$).}
 \blfootnote{\noindent H. Shin is with the Department of Electronic Engineering, Kyung Hee University, Yongin-si, Gyeonggi-do 17104, Korea (e-mail: $\rm hshin@khu.ac.kr$).}
\blfootnote{Corresponding author: Kai-Kit Wong.}

\section{Introduction}\label{sec-intro}
\subsection{Context and Literature Review}
\IEEEPARstart{I}{n the rapidly-evolving} landscape of wireless communication, the emergence of fluid antenna systems (FAS) represents a transformative leap forward \cite{new2024tutorial,Wong-2022fcn,Wang-aifas2024}. As a new form of reconfigurable antennas, FAS emphasizes on position and shape flexibility in antennas for wireless communications. FAS can be implemented by liquid-based antennas \cite{Huang-2021access}, radio-frequency reconfigurable pixels \cite{Rodrigo-2014,Jing-2022}, metamaterials \cite{Hoang-2021,Deng-2023}, mechanically movable antennas \cite{Zhu-Wong-2024} and etc. An easy way to understand the benefits of FAS is that it presents a new degree of freedom (DoF) to the physical layer that can be translated into diversity \cite{wong2020fluid} and multiplexing gains \cite{wong2021fluid}. In recent years, efforts on FAS have been spent onto improving the accuracy of channel model \cite{khammassi2023new,ramirez2024new}, and studying the diversity order in different fading channels such as Rayleigh \cite{New-2023}, Nakagami \cite{Vega-2023-2}, $\alpha$-$\mu$ distributed \cite{Alvim-2023} and even arbitrarily distributed \cite{ghadi2023copula}. Copulas were shown to be useful in the performance analysis of FAS \cite{ghadi2023gaussian}. Multiple fluid antennas at both ends were recently considered in \cite{new2023information}, revealing extraordinary diversity gains.

One interesting application for FAS is multiuser communications. In particular, a unique way of using FAS is to mitigate inter-user interference by seeking the antenna position where the interference suffers from deep fades, leading to the concept of fluid antenna multiple access (FAMA). By doing so, each receiver handles its interference by position flexibility of FAS without the requirement of precoding at the base station (BS) transmitter. The main advantage is a reduced burden at the BS for scalability in terms of both channel state information (CSI) estimation and optimization. Depending on how fast the FAS position (a.k.a.~`port') changes, FAMA can be classified into being {\em fast} \cite{wong2022fast} and {\em slow} \cite{wong2023slow,xu2024revisiting}. Opportunistic scheduling can also combine with FAMA for improved performance \cite{wong2023opportunistic,waqar2024opportunistic}. In \cite{chen2023energy}, a mean-field game was formulated to empower slow FAMA for energy-efficiency maximization.

Evidently, FAS can also be adopted as an additional DoF to elevate traditional beamforming approaches. For instance, \cite{xu2023capacity} considered an uplink scenario where the BS used traditional beamforming but the users employed multiple fluid antennas for joint port selection and beamforming in order to maximize the capacity. Moreover, FAS has been considered for integrated sensing and communications in \cite{Wang-fasisac2023,Zhou-2024,Zou-2024}. Dirty multiple access channels with FAS were also studied in \cite{ghadi2023fluid} when side information was assumed at the transmitters. The application of FAS has also extended to backscatter communications \cite{ghadi2024performance2} and wireless power transfer \cite{lin2024fluid}. Clearly, CSI is crucial to the performance of FAS and hence recent efforts have addressed the channel estimation problem for FAS in \cite{xu2023channel,Dai-2023}. Besides, machine learning has also been demonstrated to be effective in reducing the complexity for establishing the channels at the FAS ports \cite{waqar2023deep,Zhang-2024}. Finally, it is worth mentioning the recent experimental results in \cite{Shen-tap_submit2024,zhang2024pixel}, validating the capability of FAS in wireless communications systems.
	
\subsection{Motivation and Contributions}
While FAMA is attractive in simplifying how interference can be handled, it is not rate optimal or capacity achieving. In \cite{new2023achievability}, a two-user interference channel (IC) with FAS was studied and the Han-Kobayashi's achievable rate was investigated. It was revealed that FAMA without power splitting nor rate splitting could achieve a similar sum-rate to that with, if the size and resolution of FAS were sufficient. In other words, if FAS can find a port that has weak interference, then it will be optimal to treat interference as noise \cite{el2011network}. However, this is not the case when facing a strong IC in which FAS would be unable to find any port that has weak interference and therefore treating interference as noise is far from optimal.


Motivated by the above, in this paper, we aim to investigate the performance of a canonical IC where two transmitters wish to communicate to their respective receivers over a shared IC, and each of the receivers is equipped with FAS. Specifically, our interest is on the strong IC scenarios where joint decoding of both messages at both receivers is capacity optimal \cite{el2011network}. For this reason, we consider the use of simultaneous non-unique decoding (SND) \cite{Bidokhti2014}. This scheme enables receivers to decode multiple signals simultaneously, even when there is no clear separation between the desired signal and the interference.\footnote{This is particularly useful in scenarios in which users experience similar channel strengths or where interference is not suited to be effectively canceled by conventional methods such as successive interference cancellation (SIC). Unlike SIC, which decodes the strongest signal first and then cancels it out before decoding weaker ones, SND decodes multiple signals concurrently. This approach allows for more efficient use of the spectrum and expands the capacity region, especially in high-interference scenarios.} By integrating SND into the FAMA framework, we can enhance the system's ability to overcome interference, enabling each receiver to retrieve both its own message and those from the interfering transmitters without additional rate constraints. 

Our key contributions are summarized as follows:
\begin{itemize}
\item We formalize a definition of the two-user FAMA setup in the strong interference regime, and define its instantaneous capacity region. In this scenario, we derive the probability density function (PDF) and cumulative distribution function (CDF) of the signal-to-noise ratio (SNR) and interference-to-noise ratio (INR) at the receivers. 
\item Additionally, we derive closed-form expressions for the probabilistic measures such as outage probability (OP) and delay outage rate (DOR), in terms of the CDF of the multivariate normal distribution. We also provide asymptotic expressions in the high SNR regime, considering the SND interference management technique. 
\item We then derive an analytical expression for the ergodic capacity (EC) of the FAMA-IC with strong interference. This is obtained by approximating the expected value of the equivalent SNRs, i.e., the maximum of $N$ correlated random variables (RVs), using a heuristic approach.
\item Finally, we conduct numerical evaluations to assess how the implementation of FAS influences the performance of FAMA-IC under strong interference scenarios, highlighting its performance gain over the traditional fixed-position antenna system (TAS) counterpart.
\end{itemize}

\subsection{Organization}
The remainder of this manuscript is structured as follows: Section \ref{sec-sys} introduces the system model of FAMA-IC. The statistical characteristics of the equivalent channel, including the CDF and PDF, are then derived in Section \ref{sec-sta}. Section \ref{sec-per} presents the analytical results of key performance metrics such as OP, DOR, and EC. Section \ref{sec-num} provides the numerical results, and finally, Section \ref{sec-con} concludes the paper.

\begin{figure}[!t]
\centering
\includegraphics[width=0.9\columnwidth]{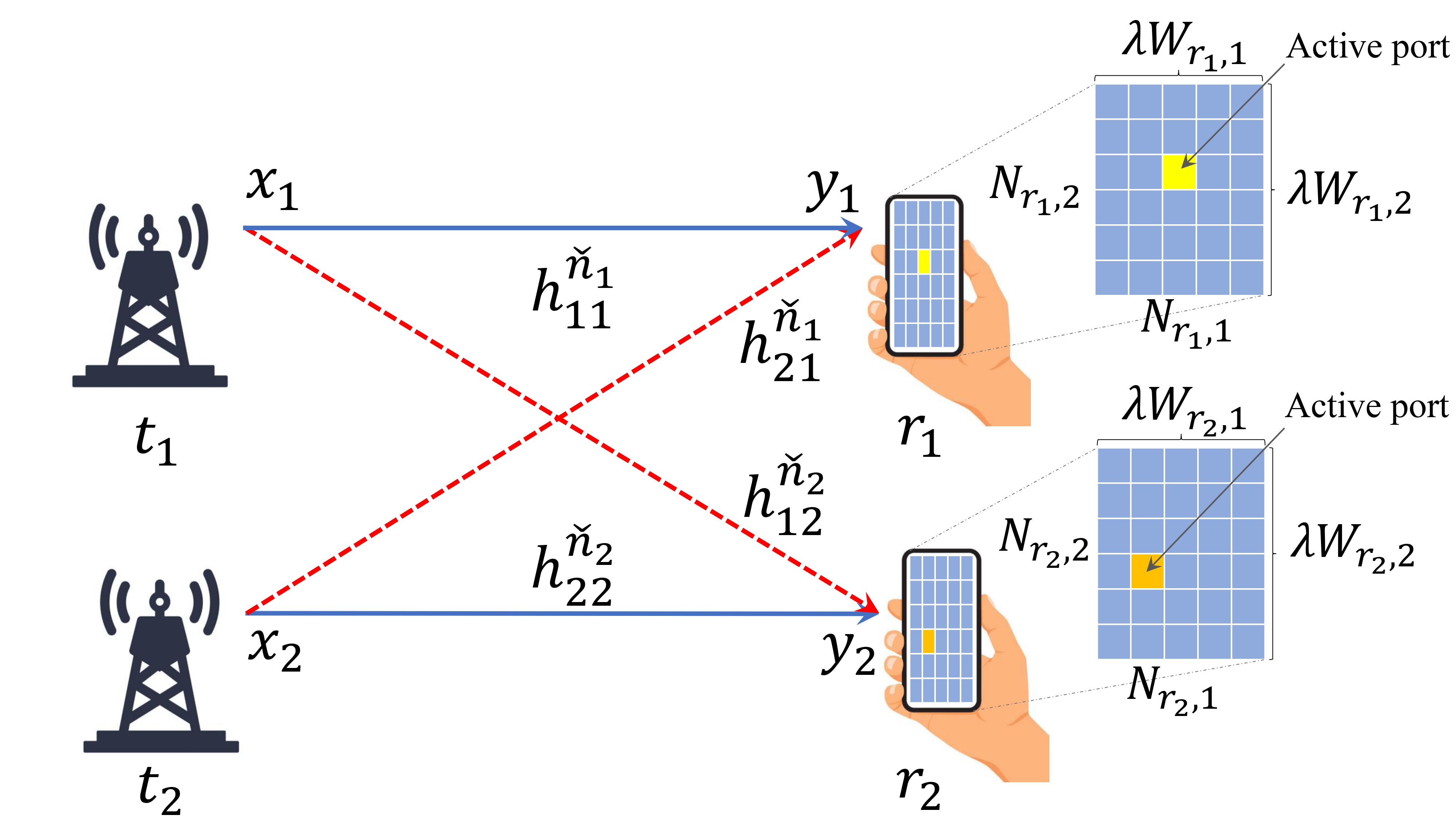}
\caption{System model for a two-user FAMA-IC.}\label{fig-model}
\end{figure}

\section{System Model}\label{sec-sys}
Consider a setup as shown in Fig.~\ref{fig-model}, where each transmitter $\mathrm{t}_i$, $i\in\left\{1,2\right\}$, sends inputs $x_i$ with transmit power $P_i$ to its respective receiver $\mathrm{r}_i$ over a shared IC. It is assumed that the transmitters $\mathrm{t}_i$ are equipped with a traditional fixed antenna while receivers $\mathrm{r}_i$ take advantage of a 2D planar fluid antenna that includes $N_{\mathrm{r}_i}$ preset ports which are uniformly distributed over an area of $W_{\mathrm{r}_i}$. To be more precise, we assume a grid structure for the fluid antenna so that $N_{{\mathrm{r}_i},l}$ ports are uniformly distributed along a linear space of length $\lambda W_{{\mathrm{r}_i},l}$ for $l\in\left\{1,2\right\}$, where $\lambda$ denotes the wavelength of the carrier frequency, i.e., $N_{\mathrm{r}_i}=N_{{\mathrm{r}_i},1}\times N_{{\mathrm{r}_i},2}$ and $W_{\mathrm{r}_i}=\lambda W_{{\mathrm{r}_i},1}\times \lambda W_{{\mathrm{r}_i},2}$. Moreover, in order to convert the $2$D indices to the $1$D index for each port (i.e., the $\left(n_{{\mathrm{r}_i},1},n_{{\mathrm{r}_i},2}\right)$-th port), we introduce an appropriate one-to-one mapping function as $\mathcal{F}\left(n_{\mathrm{r}_i}\right)=\left(n_{{\mathrm{r}_i},1},n_{{\mathrm{r}_i},2}\right)$ and $\mathcal{F}^{-1}\left(n_{{\mathrm{r}_i},1},n_{{\mathrm{r}_i},2}\right)=n_{\mathrm{r}_i}$ in which $n_{\mathrm{r}_i}\in\left\{1,\dots,N_{\mathrm{r}_i}\right\}$ and $n_{{\mathrm{r}_i},l}\in\left\{1,\dots,N_{{\mathrm{r}_i},l}\right\}$. Also, given that the fluid antenna ports can be arbitrarily close to each other, the corresponding channels are spatially correlated. Therefore, in a 3D isotropic scattering environment, the spatial correlation between the $n_{\mathrm{r}_i}$-th port and the $\tilde{n}_{\mathrm{r}_i}$-th port is given by \cite{new2023information}
\begin{align}
\varrho^{\mathrm{r}_i}_{n_{\mathrm{r}_i},\tilde{n}_{\mathrm{r}_i}}
&=\mathrm{cov}\left\{n_{\mathrm{r}_i},\tilde{n}_{\mathrm{r}_i}\right\}\notag\\
&=\mathcal{J}_0\left(2\pi\sqrt{\begin{array}{l}
\left(\frac{\vert n_{{\mathrm{r}_i},1}-\tilde{n}_{{\mathrm{r}_i},1}\vert}{N_{{\mathrm{r}_i},1}-1}W_{{\mathrm{r}_i},1}\right)^2\\
\quad\quad\quad+\left(\frac{\vert n_{{\mathrm{r}_i},2}-\tilde{n}_{{\mathrm{r}_i},2}\vert}{N_{{\mathrm{r}_i},2}-1}W_{{\mathrm{r}_i},2}\right)^2
\end{array}}\right),
\end{align}
where $\mathcal{J}_0\left(\cdot\right)$ is the spherical Bessel function of the first kind, and $\mathcal{F}\left(\tilde{n}_{\mathrm{r}_i}\right)=\left(\tilde{n}_{{\mathrm{r}_i},1},\tilde{n}_{{\mathrm{r}_i},2}\right)$ and $\mathcal{F}^{-1}\left(\tilde{n}_{{\mathrm{r}_i},1},\tilde{n}_{{\mathrm{r}_i},2}\right)=\tilde{n}_{\mathrm{r}_i}$ where  $\tilde{n}_{\mathrm{r}_i}\in\left\{1,\dots,N_{\mathrm{r}_i}\right\}$ and $\tilde{n}_{{\mathrm{r}_i},l}\in\left\{1,\dots,N_{{\mathrm{r}_i},l}\right\}$. Hence, the spatial correlation matrix $\mathbf{R}_{\mathrm{r}_i}$ is defined as
\begin{align}
\mathbf{R}_{\mathrm{r}_i}\triangleq\begin{bmatrix}
\varrho^{\mathrm{r}_i}_{1,1} & \varrho^{\mathrm{r}_i}_{1,2} &\dots& \varrho^{\mathrm{r}_i}_{1,N_{\mathrm{r}_i}}\\
\varrho^{\mathrm{r}_i}_{2,1} & \varrho^{\mathrm{r}_i}_{2,2} &\dots& \varrho^{\mathrm{r}_i}_{2,N_{\mathrm{r}_i}}\\ \vdots & \vdots & \ddots & \vdots\\
\varrho^{\mathrm{r}_i}_{N_{\mathrm{r}_i},1} & \varrho^{\mathrm{r}_i}_{N_{\mathrm{r}_i},2} &\dots& \varrho^{\mathrm{r}_i}_{N_{\mathrm{r}_i},N_{\mathrm{r}_i}}\end{bmatrix}.
\end{align}

Unlike TAS, the radiating element of each FAS receiver $\mathrm{r}_i$ can switch its location among $N_{\mathrm{r}_i}$ ports. To facilitate the notations, we define $\iota\in\left\{1,2\right\}$ as the complement of subscript $i$, so that if $i=1$, then $\iota=2$ and vice versa. Therefore, the received signal at the $n_{\mathrm{r}_i}$-th port of receiver $\mathrm{r}_i$ is given by
\begin{equation}
y^{n_{\mathrm{r}_i}}_{\mathrm{r}_i}=h^{n_{\mathrm{r}_i}}_{\mathrm{t}_i,\mathrm{r}_i}\sqrt{L_{i,i}}x_i+h^{n_{\mathrm{r}_i}}_{\mathrm{t}_{\iota},\mathrm{r}_i}\sqrt{L_{\iota,i}}x_{\iota}+z^{n_{\mathrm{r}_i}}_{\mathrm{r}_i},
\end{equation} 
where $z^{n_{\mathrm{r}_i}}_{\mathrm{r}_i}$ denotes the additive white Gaussian noise (AWGN) with zero mean and variance of $\sigma^2_i$ at the $n_{\mathrm{r}_i}$-th port of $\mathrm{r}_i$, and $L_{j,k}$ is the path loss experienced by the message from transmitter $j$ arriving at receiver $k$, for $j,k\in\{i,\iota\}$. The term $h^{n_{\mathrm{r}_i}}_{\mathrm{t}_i,\mathrm{r}_i}$ defines the corresponding normalized Rayleigh fading channel from transmitter $\mathrm{t}_i$ to receiver $\mathrm{r}_i$, i.e., ${\mathbb{E}\left\{|h^{n_{\mathrm{r}_i}}_{\mathrm{t}_i,\mathrm{r}_i}|^2\right\}=1}$, where $\mathbb{E}\left\{\cdot\right\}$ represents the expectation operator.

Specifically, we consider the scenario where each receiver $\mathrm{r}_i$ within the FAMA-IC is physically closer to the interfering transmitter $\mathrm{t}_{\iota}$ than to its own transmitter $\mathrm{t}_{i}$. This can be the case, for instance, where users are located in the expanded region of a pico-cell in the presence of a macro-cell \cite{Zhou2015}. As a result, the signal received from the interfering transmitter $\mathrm{t}_{\iota}$ surpasses that from its respective transmitter $\mathrm{t}_{i}$ in signal strength. Under such strong interference conditions, employing SND yields superior results compared to treating interference as noise.\footnote{Conventional sequential decoding schemes such as SIC typically aim to find a single best solution for decoding the desired signal at the receiver, especially when the users have different channel strengths. Conversely, SND is mainly employed when there may not be a unique solution to estimate the transmitted signal vector due to noise, interference, or channel conditions. In other words, SND is particularly useful when the interference cannot be easily cancelled or users have dissimilar yet comparable path losses.} By decoding multiple signals simultaneously using approaches such as multiuser detection or joint detection, SND can achieve higher spectral efficiency and offer a larger capacity region \cite{el2011network}. Consequently, each receiver can effectively retrieve the message of the interfering transmitter without requiring an additional constraint on its rate. With all the above definitions, the capacity region for a general two-user IC is expressed as \cite[Theorem 6.2]{el2011network}
\begin{subequations}\label{eq-capacity}
\begin{align}
R_1&<\log_2\left(1+{\rm SNR_1}\right),\\
R_2&<\log_2\left(1+{\rm SNR_2}\right),\\ \notag
R_1+R_2&<\min\left\{\log_2\left(1+{\rm SNR_1}+{\rm INR_1}\right)\right.,\\
&~\quad\quad\quad\left.\log_2\left(1+{\rm SNR_2}+{\rm INR_2}\right)\right\},
\end{align}
\end{subequations}
in which ${\rm SNR}_i$ and ${\rm INR}_i$ represent the SNR for the desired and interfering messages at user $i$, respectively. From these, the instantaneous capacity region for the considered strong interference scenario in the FAMA-IC can be defined as
\begin{subequations}\label{eq-capacity}
\begin{align}
R_1&<\log_2\left(1+\underset{n_{\mathrm{r}_1}\in\left\{1,\dots,N_{\mathrm{r}_1}\right\}}{\max}\left\{\gamma^{n_\mathrm{r_1}}_{\mathrm{r_1}}\right\}\right),\\
R_2&<\log_2\left(1+\underset{n_{\mathrm{r}_2}\in\left\{1,\dots,N_{\mathrm{r}_2}\right\}}{\max}\left\{\gamma^{n_\mathrm{r_2}}_{\mathrm{r_2}}\right\}\right),\\ \notag
R_1+R_2&<\min\left\{\log_2\left(1+\underset{\hat{n}_{\mathrm{r}_1}\in\left\{1,\dots,N_{\mathrm{r}_1}\right\}}{\max}\left\{\kappa^{\hat{n}_{\mathrm{r}_1}}_{\mathrm{r}_1}\right\}\right)\right.,\\
&~~\quad\quad\quad\left.\log_2\left(1+\underset{\hat{n}_{\mathrm{r}_2}\in\left\{1,\dots,N_{\mathrm{r}_2}\right\}}{\max}\left\{\kappa^{\hat{n}_{\mathrm{r}_2}}_{\mathrm{r}_2}\right\}\right)\right\},
\end{align}
\end{subequations}
 where $\kappa^{\hat{n}_{\mathrm{r}_i}}_{\mathrm{r}_i}=\gamma^{\hat{n}_{\mathrm{r}_i}}_{\mathrm{r}_i}+\zeta^{\hat{n}_{\mathrm{r}_i}}_{\mathrm{r}_i}$, and
\begin{equation}
\left\{\begin{aligned}
\gamma^{\check{n}_{\mathrm{r}_i}}_{\mathrm{r}_i} &= \frac{P_i L_{i,i}\left|h^{\check{n}_{\mathrm{r}_i}}_{\mathrm{t}_i\mathrm{r}_i}\right|^2}{\sigma^2_i}=\overline{\gamma}_{\mathrm{r}_i}\left|h^{\check{n}_{\mathrm{r}_i}}_{\mathrm{t}_i\mathrm{r}_i}\right|^2,\\
\zeta^{\check{n}_{\mathrm{r}_i}}_{\mathrm{r}_i} &= \frac{P_{\iota}L_{\iota,i}\left|h^{\check{n}_{\mathrm{r}_i}}_{\mathrm{t}_{\iota}\mathrm{r}_i}\right|^2}{\sigma^2_i}=\overline{\zeta}_{\mathrm{r}_i}\left|h^{\check{n}_{\mathrm{r}_i}}_{\mathrm{t}_{\iota}\mathrm{r}_i}\right|^2,
\end{aligned}\right.
\end{equation}
provided that $P_{\iota}L_{\iota,i}\vert h_{{\mathrm{t}_{
\iota}},\mathrm{r}_{i}}\vert^2>P_{i}L_{i,i}\vert h_{{\mathrm{t}_{i}},\mathrm{r}_i}\vert^2$ (strong interference condition). The terms $\overline{\gamma}_{\mathrm{r}_i}$ and $\overline{\zeta}_{\mathrm{r}_i}$ denote the average SNR and the average INR at receiver $\mathrm{r}_i$, respectively.  Note that $\check{n}\in\left\{n,\hat{n}\right\}$, $\mathcal{F}\left(\hat{n}_{\mathrm{r}_i}\right)=\left(\hat{n}_{{\mathrm{r}_i},1},\hat{n}_{{\mathrm{r}_i},2}\right)$ and $\mathcal{F}^{-1}\left(\hat{n}_{{\mathrm{r}_i},1},\hat{n}_{{\mathrm{r}_i},2}\right)=\hat{n}_{\mathrm{r}_i}$, where  $\hat{n}_{\mathrm{r}_i}\in\left\{1,\dots,N_{\mathrm{r}_i}\right\}$ and $\hat{n}_{{\mathrm{r}_i},l}\in\left\{1,\dots,N_{{\mathrm{r}_i},l}\right\}$. That is to say, we consider that only the port which maximizes the received SNR at the receiver is activated.

\section{Statistical Characterization}\label{sec-sta}
To mathematically characterize the capacity region of the FAMA-IC in \eqref{eq-capacity} under Rayleigh fading, we find it necessary to determine the distribution of the maximum amongst the $N_{\mathrm{r}_i}$ correlated SNRs, INRs, and the sum of SNRs and INRs. In this regard, we first note that all channels experience Rayleigh fading, so the corresponding channel gains are exponentially distributed, i.e., $\left|h^{n_{\mathrm{r}_i}}_{\mathrm{t}_i\mathrm{r}_i}\right|^2\sim\exp\left(1\right)$. Hence, the PDF and CDF for $\gamma_{\mathrm{r}_i}^{\check{n}_{\mathrm{r}_i}}$ and $\zeta_{\mathrm{r}_i}^{\check{n}_{\mathrm{r}_i}}$ are, respectively, expressed as
\begin{align}\label{eq-pdf-delta}
f_{\delta^{\check{n}_{\mathrm{r}_i}}_{\mathrm{r}_i}}\left(\delta_{\mathrm{r}_i}^{\check{n}_{\mathrm{r}_i}}\right)=\frac{1}{\overline{\delta}_{\mathrm{r}_i}}\exp\left(-\frac{\delta_{\mathrm{r}_i}^{\check{n}_{\mathrm{r}_i}}}{\overline{\delta}_{\mathrm{r}_i}}\right)
\end{align}
and
\begin{align}\label{eq-cdf-delta}
F_{\delta^{\check{n}_{\mathrm{r}_i}}_{\mathrm{r}_i}}\left(\delta_{\mathrm{r}_i}^{\check{n}_{\mathrm{r}_i}}\right)=1-\exp\left(-\frac{\delta_{\mathrm{r}_i}^{\check{n}_{\mathrm{r}_i}}}{\overline{\delta}_{\mathrm{r}_i}}\right),
\end{align}
in which $\delta\in\left\{\gamma,\zeta\right\}$.

Following a similar approach as in \cite{ghadi2023gaussian}, the joint multivariate distribution for the maximum of $d$ correlated arbitrary RVs can be accurately generated by using Sklar's theorem. More specifically, the joint CDF of $d$ RVs $S_1,\dots, S_d$ with the corresponding marginal CDF $F_{S_l}\left(s_l\right)$, $l\in\left\{1,\dots,d\right\}$, in the extended real line domain $\mathbb{R}$, is given by
\begin{align}\label{eq-sklar}
F_{S_1,\dots,S_d}\left(s_1,\dots,s_d\right)=C\left(F_{S_1}\left(s_1\right),\dots,F_{S_d}\left(s_d\right);\vartheta\right),
\end{align}
where $C:\left[0,1\right]^d\rightarrow\left[0,1\right]$ is the $d$-dimensional copula, which is defined as a joint CDF of $d$ RVs on the unit cube $\left[0,1\right]^d$ with uniform marginal distributions. Besides, $\vartheta$ is the copula dependence parameter that measures the linear and/or non-linear correlation between the underling RVs. In this regard, Gaussian copula is one of the practical models to describe the inherent spatial correlation between fluid antenna ports due to its appealing properties \cite{ghadi2023gaussian}. Thus, by defining the maximum of  $\gamma_{\mathrm{r}_i}^{n_{\mathrm{r}_i}}$ and $\zeta_{\mathrm{r}_i}^{n_{\mathrm{r}_i}}$ as 
\begin{align}
\delta^\mathrm{fama}_{\mathrm{r}_i}=\max\left\{\delta_{\mathrm{r}_i}^{1},\dots,\delta_{\mathrm{r}_i}^{N_{\mathrm{r}_i}}\right\}, ~\delta\in\left\{\gamma,\zeta\right\},
\end{align}
the CDF of $\delta^\mathrm{fama}_{\mathrm{r}_i}$ can be found as \cite{ghadi2023gaussian}
\begin{align}\nonumber
&F_{\delta^\mathrm{fama}_{\mathrm{r}_i}}\left(r\right)\\
&=\Phi_{\mathbf{R}_{\mathrm{r}_i}}\left(\varphi^{-1}\left(F_{\delta_{\mathrm{r}_i}^\mathrm{1}}\left(r\right)\right),\dots,\varphi^{-1}\left(F_{\delta_{\mathrm{r}_i}^{N_{\mathrm{r}_i}}}\left(r\right)\right);\varpi^{\mathrm{r}_i}\right),\label{eq-cdf}
\end{align}
in which $F_{\delta_{\mathrm{r}_i}^{n_{\mathrm{r}_i}}}\left(r\right)$ is defined in \eqref{eq-cdf-delta}, $\Phi_\mathbf{R}(\cdot)$ is the joint CDF of the multivariate normal distribution with zero mean vector and correlation matrix $\mathbf{R}_{\mathrm{r}_i}$, 
$\varphi^{-1}\left(x\right)=\sqrt{2}\mathrm{erf}^{-1}\left(2x-1\right)$ is the quantile function of the standard normal distribution, where $\mathrm{erf}^{-1}\left(\cdot\right)$ is the inverse of the error function $\mathrm{erf}\left(z\right)=\frac{2}{\sqrt{\pi}}\int_0^z\mathrm{e}^{-t^2}dt$. In addition, $\varpi^{\mathrm{r}_i}$ defines the dependence parameter of the Gaussian copula. It is noted that by exploiting the Cholesky decomposition, we have $\varpi^{\mathrm{r}_i}_{n_{\mathrm{r}_i},\tilde{n}_{\mathrm{r}_i}}\approx\varrho^{\mathrm{r}_i}_{n_{\mathrm{r}_i},\tilde{n}_{\mathrm{r}_i}}$ \cite{ghadi2023gaussian}. 

By applying the chain rule to \eqref{eq-cdf}, the PDF of $\delta^\mathrm{fama}_{\mathrm{r}_i}$ is given by
\begin{multline}\label{eq-pdf-gaussian}
f_{\delta^\mathrm{fama}_{\mathrm{r}_i}}\left(r\right)\\
=\prod_{n_{\mathrm{r}_i}=1}^{K_{\mathrm{r}_i}}f_{\delta_{\mathrm{r}_i}}^{n_{\mathrm{r}_i}}\left(r\right)\frac{\exp\left(-\frac{1}{2}\left(\boldsymbol{\vec{\varphi}}^{-1}_{\delta_{\mathrm{r}_i}}\right)^T\left(\mathbf{R}_{\mathrm{r}_i}^{-1}-\mathbf{I}\right)\boldsymbol{\vec{\varphi}}^{-1}_{\delta_{\mathrm{r}_i}}\right)}{\sqrt{{\rm det}\left(\mathbf{R}_{\mathrm{r}_i}\right)}},
\end{multline}
where $f_{\delta_{\mathrm{r}_i}^{n_{\mathrm{r}_i}}}\left(r\right)$ is given in \eqref{eq-pdf-delta}, $\mathrm{det}\left(\mathbf{R}_{\mathrm{r}_i}\right)$ is the determinant of the correlation matrix $\mathbf{R}_{\mathrm{r}_i}$, $\mathbf{I}$ is an identity matrix, and 
\begin{align}
\boldsymbol{\vec{\varphi}}^{-1}_{\delta_{\mathrm{r}_i}}=\left[\varphi^{-1}\left(F_{\delta_{\mathrm{r}_i}^\mathrm{1}}\left(r\right)\right),\dots,\varphi^{-1}\left(F_{\delta_{\mathrm{r}_i}^{N_{\mathrm{r}_i}}}\left(r\right)\right)\right]^T.
\end{align}

Now, we need to determine the distribution corresponding to the maximum of the sum of the SNR and INR, i.e., $\kappa^{\hat{n}_{\mathrm{r}_i}}_{\mathrm{r}_i}=\gamma^{\hat{n}_{\mathrm{r}_i}}_{\mathrm{r}_i}+\zeta^{\hat{n}_{\mathrm{r}_i}}_{\mathrm{r}_i}$, and we define 
\begin{align}
\kappa_{\mathrm{r}_i}^\mathrm{fama}=\max\left\{\kappa^1_{\mathrm{r}_i},\dots,\kappa^{N_{\mathrm{r}_i}}_{\mathrm{r}_i}\right\}. 
\end{align}
Given that $\kappa^{\hat{n}_{\mathrm{r}_i}}_{\mathrm{r}_i}$ includes the sum of two independent exponential RVs with different scale parameters $\overline{\gamma}_{\mathrm{r}_i}$ and $\overline{\zeta}_{\mathrm{r}_i}$, it has the hypoexponential distribution with the following PDF and CDF \cite{ghadi2022copula}:
\begin{align}\label{eq-pdf-kappa}
f_{\kappa_{\mathrm{r}_i}^{\hat{n}_{\mathrm{r}_i}}}\left(\kappa_{\mathrm{r}_i}^{\hat{n}_{\mathrm{r}_i}}\right)=\frac{1}{\Delta_{\mathrm{r}_i}}\left(\mathrm{e}^{-\frac{\kappa_{\mathrm{r}_i}^{\hat{n}_{\mathrm{r}_i}}}{\overline{\gamma}_{\mathrm{r}_i}}}-\mathrm{e}^{-\frac{\kappa_{\mathrm{r}_i}^{\hat{n}_{\mathrm{r}_i}}}{\overline{\zeta}_{\mathrm{r}_i}}}\right)
\end{align}
and
\begin{align}\label{eq-cdf-kappa}
F_{\kappa_{\mathrm{r}_i}^{\hat{n}_{\mathrm{r}_i}}}\left(\kappa_{\mathrm{r}_i}^{\hat{n}_{\mathrm{r}_i}}\right)=1-\frac{1}{\Delta_{\mathrm{r}_i}}\left(\overline{\gamma}_{\mathrm{r}_i}\mathrm{e}^{-\frac{\kappa_{\mathrm{r}_i}^{\hat{n}_{\mathrm{r}_i}}}{\overline{\gamma}_{\mathrm{r}_i}}}-\overline{\zeta}_{\mathrm{r}_i}\mathrm{e}^{-\frac{\kappa_{\mathrm{r}_i}^{\hat{n}_{\mathrm{r}_i}}}{\overline{\zeta}_{\mathrm{r}_i}}}\right),
\end{align}
respectively, where $\Delta_{\mathrm{r}_i}=\overline{\gamma}_{\mathrm{r}_i}-\overline{\zeta}_{\mathrm{r}_i}$. 

As such, using Sklar's theorem in \eqref{eq-sklar} and the Gaussian copula definition, the CDF and PDF of $\kappa_{\mathrm{r}_i}^\mathrm{fama}$ are found as
\begin{align}
&F_{\kappa^\mathrm{fama}_{\mathrm{r}_i}}\left(r\right)\notag\\
&=\Phi_{\mathbf{R}_{\mathrm{r}_i}}\left(\varphi^{-1}\left(F_{\kappa_{\mathrm{r}_i}^\mathrm{1}}\left(r\right)\right),\dots,\varphi^{-1}\left(F_{\kappa_{\mathrm{r}_i}^{N_{\mathrm{r}_i}}}\left(r\right)\right);\varpi^{\mathrm{r}_i}\right),\label{eq-cdf-kappa1}
\end{align}
and
\begin{multline}\label{eq-pdf-gaussian-kappa}
f_{\kappa^\mathrm{fama}_{\mathrm{r}_i}}\left(r\right)\\
=\prod_{\hat{n}_{\mathrm{r}_i}=1}^{N_{\mathrm{r}_i}}f_{\kappa_{\mathrm{r}_i}}^{\hat{n}_{\mathrm{r}_i}}\left(r\right)\frac{\exp\left(-\frac{1}{2}\left(\boldsymbol{\vec{\varphi}}^{-1}_{\kappa_{\mathrm{r}_i}}\right)^T\left(\mathbf{R}_{\mathrm{r}_i}^{-1}-\mathbf{I}\right)\boldsymbol{\vec{\varphi}}^{-1}_{\kappa_{\mathrm{r}_i}}\right)}{\sqrt{{\rm det}\left(\mathbf{R}_{\mathrm{r}_i}\right)}},
\end{multline}
in which $F_{\kappa_{\mathrm{r}_i}^{\hat{n}_{\mathrm{r}_i}}}\left(r\right)$ and $f_{\kappa_{\mathrm{r}_i}^{\hat{n}_{\mathrm{r}_i}}}\left(r\right)$ are defined in \eqref{eq-cdf-kappa} and \eqref{eq-pdf-kappa}, respectively. Moreover, $\boldsymbol{\vec{\varphi}}^{-1}_{\kappa_{\mathrm{r}_i}}$ is denoted as
\begin{align}
\boldsymbol{\vec{\varphi}}^{-1}_{\kappa_{\mathrm{r}_i}}\triangleq\left[\varphi^{-1}\left(F_{\kappa_{\mathrm{r}_i}^\mathrm{1}}\left(r\right)\right),\dots,\varphi^{-1}\left(F_{\kappa_{\mathrm{r}_i}^{N_{\mathrm{r}_i}}}\left(r\right)\right)\right]^T.
\end{align}
Note that the copula approach allows for a tractable characterization for the joint distribution of $\kappa_{\mathrm{r}_i}^\mathrm{fama}$, from their underlying marginal distributions. Having characterized the key statistics of the equivalent channels under consideration, these will be used to analyze a number of performance metrics.

\section{Performance Analysis}\label{sec-per}
\subsection{OP}
OP is an important performance metric in evaluating wireless communication systems, and is defined as the probability that the capacity falls below a given rate threshold $R_\mathrm{th}$, i.e., $P_\mathrm{out}={\rm Pr}\left(C\leq R_\mathrm{th}\right)$, where $C=\log_2\left(1+\delta\right)$. Hence, for the considered FAMA-IC, an outage occurs if at least one of the following events happens:
 \begin{subequations}
 \begin{align}
\mathcal{E}_1:&\log_2\left(1+\gamma_{\mathrm{r}_1}^{\mathrm{fama}}\right)\leq R_1^\mathrm{th},\\
\mathcal{E}_2:&\log_2\left(1+\gamma_{\mathrm{r}_2}^{\mathrm{fama}}\right)\leq R_2^\mathrm{th},\\
\mathcal{E}_3:&\min\left\{ \log_2\left(1+\kappa_{\mathrm{r}_1}^{\mathrm{fama}}\right),\log_2\left(1+\kappa_{\mathrm{r}_2}^{\mathrm{fama}}\right)\right\}\leq R_\mathrm{sum}^\mathrm{th},
\end{align}
\end{subequations}
where $R_i^\mathrm{th}$ and  $R_\mathrm{sum}^\mathrm{th}=R_1^\mathrm{th}+R_2^\mathrm{th}$ are the rate threshold and the sum-rate threshold, respectively. Therefore, the OP for the considered FAMA-IC is defined as 
\begin{align}\label{eq-op-def}
P_\mathrm{out}=\Pr\left(\mathcal{E}_1\cup \mathcal{E}_2\cup \mathcal{E}_3\right).
\end{align}

\begin{theorem}
The OP for the FAMA-IC system under the strong interference scenario is given by
\begin{multline}
P_\mathrm{out}=1-
\overline{F}_{\kappa_\mathrm{r_1}^\mathrm{fama}}\left(\tilde{R}^\mathrm{th}_\mathrm{sum}\right)\times\\
\overline{F}_{\kappa_\mathrm{r_2}^\mathrm{fama}}\left(\tilde{R}^\mathrm{th}_\mathrm{sum}\right)\overline{F}_{\gamma_{\mathrm{r}_1}^\mathrm{fama},\gamma^\mathrm{fama}_{\mathrm{r}_2}}\left(\tilde{R}^\mathrm{th}_1,\tilde{R}^\mathrm{th}_2\right),
\end{multline}
in which $\overline{F}_{\kappa_{\mathrm{r}_i}^\mathrm{fama}}\left(\tilde{R}^\mathrm{th}_\mathrm{sum}\right)$ represents the complementary CDF (CCDF) of $\kappa_{\mathrm{r}_i}^\mathrm{fama}$,  $\overline{F}_{\gamma_{\mathrm{r}_1}^\mathrm{fama},\gamma^\mathrm{fama}_{\mathrm{r}_2}}\left(\tilde{R}^\mathrm{th}_1,\tilde{R}^\mathrm{th}_2\right)$ is the joint CCDF of $\gamma_{\mathrm{r}_1}^\mathrm{fama}$ and $\gamma_{\mathrm{r}_2}^\mathrm{fama}$, $\tilde{R}^\mathrm{th}_i=2^{R^\mathrm{th}_i}-1$, and $\tilde{R}_\mathrm{sum}^\mathrm{th}=2^{R^\mathrm{th}_\mathrm{sum}}-1$.
\end{theorem}

\begin{proof}
The OP definition in \eqref{eq-op-def} can be extended as
\begin{multline}
P_\mathrm{out}=\Pr\left(\mathcal{E}_1\right)+\Pr\left(\mathcal{E}_2\right)+\Pr\left(\mathcal{E}_3\right)\\
-\Pr\left(\mathcal{E}_1\cap\mathcal{E}_2\right)-\Pr\left(\mathcal{E}_1\cap\mathcal{E}_3\right)-\Pr\left(\mathcal{E}_2\cap\mathcal{E}_3\right)\\
+\Pr\left(\mathcal{E}_1\cap\mathcal{E}_2\cap\mathcal{E}_3\right),
\end{multline}
where $\Pr\left(\mathcal{E}_1\right)$ can be determined as
\begin{subequations}
\begin{align}
\Pr\left(\mathcal{E}_1\right) &= \Pr\left(\log_2\left(1+\gamma_{\mathrm{r}_1}^{\mathrm{fama}}\right)\leq R_1^\mathrm{th}\right)\\
&=\Pr\left(\gamma_{\mathrm{r}_1}^{\mathrm{fama}}\leq\tilde{R}^\mathrm{th}_1\right)\\
& = F_{\gamma^\mathrm{fama}_{\mathrm{r}_1}}\left(\tilde{R}^\mathrm{th}_1\right),
\end{align}
\end{subequations}
where $\tilde{R}^\mathrm{th}_1=2^{R^\mathrm{th}_1}-1$. Similarly, $\Pr\left(\mathcal{E}_2\right)= F_{\gamma^\mathrm{fama}_{\mathrm{r}_2}}\left(\tilde{R}^\mathrm{th}_2\right)$. As for the event $\mathcal{E}_3$, we have
\begin{subequations}
\begin{align}\notag
&\Pr\left(\mathcal{E}_3\right)=\\
&\Pr\left(\min\left\{\log_2\left(1+\kappa_{\mathrm{r}_1}^{\mathrm{fama}}\right),\log_2\left(1+\kappa_{\mathrm{r}_2}^{\mathrm{fama}}\right)\right\}\hspace{-1mm}\leq R_\mathrm{sum}^\mathrm{th}\right)\\\notag
&=1-\Pr\left(\log_2\left(1+\kappa_{\mathrm{r}_1}^{\mathrm{fama}}\right)>R_\mathrm{sum}^\mathrm{th}\right)\\ \label{eq-proof-b}
&\quad\quad\quad\times \Pr\left(\log_2\left(1+\kappa_{\mathrm{r}_1}^{\mathrm{fama}}\right)>R_\mathrm{sum}^\mathrm{th}\right)\\ 
&=1-\Pr\left(\kappa_{\mathrm{r}_1}^{\mathrm{fama}}>\tilde{R}_\mathrm{sum}^\mathrm{th}\right)\Pr\left(\kappa_{\mathrm{r}_2}^{\mathrm{fama}}>\tilde{R}_\mathrm{sum}^\mathrm{th}\right)\\
&=1-\overline{F}_{\kappa_\mathrm{r_1}^\mathrm{fama}}\left(\tilde{R}^\mathrm{th}_\mathrm{sum}\right)\overline{F}_{\kappa_\mathrm{r_2}^\mathrm{fama}}\left(\tilde{R}^\mathrm{th}_\mathrm{sum}\right),
\end{align}
\end{subequations}
where \eqref{eq-proof-b} is obtained from the independence of the events, $\tilde{R}_\mathrm{sum}^\mathrm{th}=2^{R^\mathrm{th}_\mathrm{sum}}-1$, and $\overline{F}_{\kappa_{\mathrm{r}_i}^\mathrm{fama}}\left(r\right)=1-F_{\kappa_{\mathrm{r}_i}^\mathrm{fama}}\left(r\right)$ is the CCDF of $\kappa_{\mathrm{r}_i}^\mathrm{fama}$. Since all events $\mathcal{E}_j$ for $j\in\left\{1,2,3\right\}$ are independent of each other, the other probabilities can be obtained according to $\Pr\left(\mathcal{E}_j\right)$. Therefore, by considering the CDF of $\gamma^\mathrm{fama}_{\mathrm{r}_i}$, $\zeta^\mathrm{fama}_{\mathrm{r}_i}$, and $\kappa^\mathrm{fama}_{\mathrm{r}_i}$, and after some mathematical simplification, the proof is accomplished.  
\end{proof}

From the derived closed-form expression for OP, we now aim to understand the asymptotic behavior of the OP at high SNR, i.e., $\gamma_{{\mathrm{r}}_i}\rightarrow\infty$ while meeting the strong interference condition. As we will later see, this asymptotic expression provides more insights into the system's reliability.

\begin{corollary}\label{col-out}
The asymptotic OP in the high SNR regime for the FAMA-IC system under the strong interference scenario is given by
\begin{multline}
P_\mathrm{out}^\infty=1-
\overline{F}_{\kappa_\mathrm{r_1}^\mathrm{fama}}^\infty\left(\tilde{R}^\mathrm{th}_\mathrm{sum}\right)\\
\times\overline{F}_{\kappa_\mathrm{r_2}^\mathrm{fama}}^\infty\left(\tilde{R}^\mathrm{th}_\mathrm{sum}\right)\overline{F}_{\gamma_{\mathrm{r}_1}^\mathrm{fama},\gamma^\mathrm{fama}_{\mathrm{r}_2}}^\infty\left(\tilde{R}^\mathrm{th}_1,\tilde{R}^\mathrm{th}_2\right),
\end{multline}
where $\overline{F}_{\kappa_{\mathrm{r}_i}^\mathrm{fama}}^\infty\left(\tilde{R}^\mathrm{th}_\mathrm{sum}\right)$ and $\overline{F}_{\gamma_{\mathrm{r}_1}^\mathrm{fama},\gamma^\mathrm{fama}_{\mathrm{r}_2}}^\infty\left(\tilde{R}^\mathrm{th}_1,\tilde{R}^\mathrm{th}_2\right)$ represent the CCDF of $\kappa_{\mathrm{r}_i}^\mathrm{fama}$ and the joint CCDF of $\gamma_{\mathrm{r}_1}^\mathrm{fama}$ and $\gamma_{\mathrm{r}_2}^\mathrm{fama}$ in the high SNR regime, provided in \eqref{eq-cdf1-asym} and \eqref{eq-cdf2-asym}, respectively. 
\begin{figure*}[!t]
\begin{align}
	\overline{F}_{\kappa_{\mathrm{r}_i}^\mathrm{fama}}^\infty\left(\tilde{R}^\mathrm{th}_\mathrm{sum}\right)=1-
	\Phi_{\mathbf{R}_i}\left(\varphi^{-1}\left(\frac{1}{\Delta_{\mathrm{r}_i}}\left(\tilde{R}^\mathrm{th}_\mathrm{sum}+\overline{\zeta}_{\mathrm{r}_i}\mathrm{e}^{-\frac{\tilde{R}^\mathrm{th}_\mathrm{sum}}{\overline{\zeta}_{\mathrm{r}_i}}}\right)\right),\dots,\varphi^{-1}\left(\frac{1}{\Delta_{\mathrm{r}_i}}\left(\tilde{R}^\mathrm{th}_\mathrm{sum}+\overline{\zeta}_{\mathrm{r}_i}\mathrm{e}^{-\frac{\tilde{R}^\mathrm{th}_\mathrm{sum}}{\overline{\zeta}_{\mathrm{r}_i}}}\right)\right);\varpi^{\mathrm{r}_i}\right)\label{eq-cdf1-asym}
\end{align}
\hrulefill
\begin{align}\notag
&\overline{F}_{\gamma_{\mathrm{r}_1}^\mathrm{fama},\gamma^\mathrm{fama}_{\mathrm{r}_2}}^\infty\left(\tilde{R}^\mathrm{th}_1,\tilde{R}^\mathrm{th}_2\right)\\
&\quad=\Bigg[1- \Phi_{\mathbf{R}_{\mathrm{r}_1}}\left(\varphi^{-1}\left(\frac{\tilde{R}^\mathrm{th}_1}{\overline{\gamma}_{\mathrm{r}_1}}\right),\dots,\varphi^{-1}\left(\frac{\tilde{R}^\mathrm{th}_1}{\overline{\gamma}_{\mathrm{r}_1}}\right);\varpi^{\mathrm{r}_1}\right)\Bigg] \Bigg[1- \Phi_{\mathbf{R}_{\mathrm{r}_2}}\left(\varphi^{-1}\left(\frac{\tilde{R}^\mathrm{th}_2}{\overline{\gamma}_{\mathrm{r}_2}}\right),\dots,\varphi^{-1}\left(\frac{\tilde{R}^\mathrm{th}_2}{\overline{\gamma}_{\mathrm{r}_2}}\right);\varpi^{\mathrm{r}_2}\right)\Bigg]\label{eq-cdf2-asym}
\end{align}
\hrulefill
\end{figure*}
\end{corollary}

\begin{proof}
Given that $\gamma^{n_{\mathrm{r}_i}}_{\mathrm{r}_i}$ and $	\kappa^{\hat{n}_{\mathrm{r}_i}}_{\mathrm{r}_i}$ follow the exponential and hypoexponential distributions, respectively, their corresponding CDFs in the high SNR regime can be derived with the help of the Taylor series expansion. Therefore, as $\overline{\gamma}_{\mathrm{r}_i}\rightarrow\infty$, the term $\mathrm{e}^{-{\tilde{R}^\mathrm{th}_\nu}/\overline{\gamma}_{\mathrm{r}_i}}$ for $\nu\in\left\{1,2,\textrm{sum}\right\}$ becomes very small such that we have
\begin{align}
\exp\left(-\frac{{\tilde{R}^\mathrm{th}_\nu}}{\overline{\gamma}_{\mathrm{r}_i}}\right)\approx 1-\frac{\tilde{R}^\mathrm{th}_\nu}{\overline{\gamma}_{\mathrm{r}_i}}, \quad \nu\in\left\{1,2,\textrm{sum}\right\}.\label{eq-exp}
\end{align}
By substituting \eqref{eq-exp} into \eqref{eq-cdf-delta} and \eqref{eq-cdf-kappa}, the marginal CDFs of   $\gamma_{\mathrm{r}_i}^{n_{\mathrm{r}_i}}$ and $\kappa_{\mathrm{r}_i}^{\hat{n}_{\mathrm{r}_i}}$ at high SNR are, respectively, given by
\begin{align}\label{eq-cdf-kappa-inf}
F_{\kappa_{\mathrm{r}_i}^{\hat{n}_{\mathrm{r}_i}}}^\infty\left(\tilde{R}^\mathrm{th}_\mathrm{sum}\right)=\frac{1}{\Delta_{\mathrm{r}_i}}\left(\tilde{R}^\mathrm{th}_\mathrm{sum}+\overline{\zeta}_{\mathrm{r}_i}\exp\left(-\frac{\tilde{R}^\mathrm{th}_\mathrm{sum}}{\overline{\zeta}_{\mathrm{r}_i}}\right)\right),
\end{align}
and
\begin{align}\label{eq-cdf-gamma-inf}
F_{\gamma^{\check{n}_{\mathrm{r}_i}}_{\mathrm{r}_i}}^\infty\left(\tilde{R}^\mathrm{th}_i\right)=\frac{\tilde{R}^\mathrm{th}_i}{\overline{\gamma}_{\mathrm{r}_i}}.
\end{align}
Now, by inserting \eqref{eq-cdf-gamma-inf} into \eqref{eq-cdf} and plugging \eqref{eq-cdf-kappa-inf} into \eqref{eq-cdf-kappa1}, $\overline{F}_{\kappa_{\mathrm{r}_i}^\mathrm{fama}}^\infty\left(\tilde{R}^\mathrm{th}_\mathrm{sum}\right)$ and $\overline{F}_{\gamma_{\mathrm{r}_1}^\mathrm{fama},\gamma^\mathrm{fama}_{\mathrm{r}_2}}^\infty\left(\tilde{R}^\mathrm{th}_1,\tilde{R}^\mathrm{th}_2\right)$ are obtained and the proof becomes complete. 
\end{proof}

\subsection{DOR}
The DOR is a crucial metric in the design of various wireless communication scenarios, particularly in ultra-reliable and low-latency communication systems. It is defined as the probability that the time required to transmit a specific amount of data $R$ over a wireless channel with bandwidth $B$ exceeds a predefined threshold time $T_\mathrm{th}$ \cite{Yang2019}. Mathematically, this is expressed as $\Pr\left(T_\mathrm{d}>T_\mathrm{th}\right)$, where $T_\mathrm{d}=\frac{R}{BC}$ indicates the delivery time such that $C$ denotes the channel capacity \cite{ghadi2024performance2}. This expression captures the essential relationship between data transmission time, channel bandwidth, and signal quality.

Thus, for the considered FAMA-IC, the DOR is achieved if at least one of the following events occurs:
\begin{subequations}
\begin{align}
\mathcal{T}_1:& \frac{R_1}{B_1C_1}> T_1^\mathrm{th},\\
\mathcal{T}_2:& \frac{R_2}{B_2C_2}> T_2^\mathrm{th},\\
\mathcal{T}_3:& \frac{R_1+R_2}{B_\mathrm{sum}C_\mathrm{sum}}> T_\mathrm{sum}^\mathrm{th},
\end{align}
\end{subequations}
where $C_\nu$ for $\nu\in\left\{1,2,\mathrm{sum}\right\}$ is defined in \eqref{eq-capacity}, and $B_\nu$ and $R_i$ denote the corresponding bandwidth and rates, respectively. Therefore, the DOR for FAMA-IC is defined as
\begin{align}\label{eq-dor-def}
	P_\mathrm{dor}=\Pr\left(\mathcal{T}_1\cup \mathcal{T}_2\cup \mathcal{T}_3\right).
\end{align}
\begin{theorem}
The OP for the FAMA-IC system under the strong interference scenario is given by
\begin{multline}
P_\mathrm{dor}=1-
\overline{F}_{\kappa_\mathrm{r_1}^\mathrm{fama}}\left(\hat{T}^\mathrm{th}_\mathrm{sum}\right)\\
\times\overline{F}_{\kappa_\mathrm{r_2}^\mathrm{fama}}\left(\hat{T}^\mathrm{th}_\mathrm{sum}\right)\overline{F}_{\gamma_{\mathrm{r}_1}^\mathrm{fama},\gamma^\mathrm{fama}_{\mathrm{r}_2}}\left(\hat{T}^\mathrm{th}_1,\hat{T}^\mathrm{th}_2\right),
\end{multline}
in which $\hat{T}_i=\mathrm{e}^{\frac{R_i\ln 2}{BT_i^\mathrm{th}}}-1$ and $\hat{T}_\mathrm{sum}=\mathrm{e}^{\frac{\left(R_1+R_2\right)\ln 2}{B_\mathrm{sum}T_\mathrm{sum}^\mathrm{th}}}$.
\end{theorem}

\begin{proof}
By extending the definition of DOR in \eqref{eq-dor-def}, we have
\begin{multline}
P_\mathrm{dor}=\Pr\left(\mathcal{T}_1\right)+\Pr\left(\mathcal{T}_2\right)+\Pr\left(\mathcal{T}_3\right)-\Pr\left(\mathcal{T}_1\cap\mathcal{T}_2\right)\\
-\Pr\left(\mathcal{T}_1\cap\mathcal{T}_3\right)-\Pr\left(\mathcal{T}_2\cap\mathcal{T}_3\right)+\Pr\left(\mathcal{T}_1\cap\mathcal{T}_2\cap\mathcal{T}_3\right),
\end{multline}
where $\Pr\left(\mathcal{T}_1\right)$ can be mathematically obtained as
\begin{subequations}
\begin{align}
\Pr\left(\mathcal{T}_1\right) &= \Pr\left(\frac{R_1}{B_1\log_2\left(1+\gamma_{\mathrm{r}_1}^{\mathrm{fama}}\right)}> T_1^\mathrm{th}\right)\\
& = \Pr\left(\gamma_{\mathrm{r}_1}^\mathrm{fama}\leq \exp\left(\frac{R_1\ln 2}{B_1T_1^\mathrm{th}}\right)-1\right)\\
&=F_{\gamma_{\mathrm{r}_1}^\mathrm{fama}}\left(\hat{T}_1\right)
\end{align}
\end{subequations}
where $\hat{T}_1=\mathrm{e}^{\frac{R_1\ln 2}{B_1T_1^\mathrm{th}}}-1$. Similarly, $\Pr\left(\mathcal{T}_2\right)=F_{\gamma_{\mathrm{r}_2}^\mathrm{fama}}\left(\hat{T}_2\right)$ in which $\hat{T}_2=\mathrm{e}^{\frac{R_2\ln 2}{B_1T_2^\mathrm{th}}}-1$. Moreover, $\Pr\left(\mathcal{T}_3\right)$ is derived as
\begin{subequations}
\begin{align}\notag
&\Pr\left(\mathcal{T}_3\right) =\\
& \Pr\hspace{-1mm}\left(\hspace{-1mm}\tfrac{\left(R_1+R_2\right)}{B_\mathrm{sum}\min\left\{\log_2\left(1+\kappa_{\mathrm{r}_1}^{\mathrm{fama}}\right),\log_2\left(1+\kappa_{\mathrm{r}_2}^{\mathrm{fama}}\right)\right\}}\hspace{-1mm}>\hspace{-1mm} T_\mathrm{sum}^\mathrm{th}\hspace{-1.5mm}\right)\\\notag
& = 1 - \Pr\bigg(\min\left\{\log_2\left(1+\kappa_{\mathrm{r}_1}^{\mathrm{fama}}\right),\log_2\left(1+\kappa_{\mathrm{r}_2}^{\mathrm{fama}}\right)\right\}\\
&\hspace{4cm}<\frac{\left(R_1+R_2\right)\ln 2}{B_\mathrm{sum}T_\mathrm{sum}^\mathrm{th}}\bigg)\\\notag
& = 1-\Pr\left(\kappa_{\mathrm{r}_1}>\exp\left(\frac{\left(R_1+R_2\right)\ln 2}{B_\mathrm{sum}T_\mathrm{sum}^\mathrm{th}}\right)-1\right)\\
&\hspace{1.2cm}\times \Pr\left(\kappa_{\mathrm{r}_2}>\exp\left(\frac{\left(R_1+R_2\right)\ln 2}{B_\mathrm{sum}T_\mathrm{sum}^\mathrm{th}}\right)-1\right)\\
& = 1-\overline{F}_{\kappa_\mathrm{r_1}^\mathrm{fama}}\left(\hat{T}_\mathrm{sum}\right)\overline{F}_{\kappa_\mathrm{r_2}^\mathrm{fama}}\left(\hat{T}_\mathrm{sum}\right), 
\end{align}
\end{subequations}
where $\hat{T}_\mathrm{sum}=\mathrm{e}^{\frac{2\left(R_1+R_2\right)\ln 2}{B_\mathrm{sum}T_\mathrm{sum}^\mathrm{th}}}$. Now, given the independence of all events $\mathcal{T}_j$ for $j\in\left\{1,2,3\right\}$, other probabilities can be obtained according to $\Pr\left(\mathcal{T}_j\right)$. Hence, by considering the CDFs of $\gamma^\mathrm{fama}_{\mathrm{r}_i}$, $\zeta^\mathrm{fama}_{\mathrm{r}_i}$, and $\kappa^\mathrm{fama}_{\mathrm{r}_i}$, and after some mathematical simplification, the proof is completed.
\end{proof}

\begin{corollary}
The asymptotic DOR in the high SNR regime for the FAMA-IC under the strong interference scenario is given by
\begin{multline}
P_\mathrm{dor}^\infty=1-
\overline{F}_{\kappa_\mathrm{r_1}^\mathrm{fama}}^\infty\left(\hat{T}^\mathrm{th}_\mathrm{sum}\right)\\
\times\overline{F}_{\kappa_\mathrm{r_2}^\mathrm{fama}}^\infty\left(\hat{T}^\mathrm{th}_\mathrm{sum}\right)\overline{F}_{\gamma_{\mathrm{r}_1}^\mathrm{fama},\gamma^\mathrm{fama}_{\mathrm{r}_2}}^\infty\left(\hat{T}^\mathrm{th}_1,\hat{T}^\mathrm{th}_2\right).
\end{multline}
\end{corollary}

\begin{proof}
Following the same approach used to prove Corollary \ref{col-out}, the proof is completed. 
\end{proof}

\subsection{EC}
The EC (per bandwidth unit) for the system model under consideration is defined as
\begin{align}
\overline{\mathcal{C}}\left[\mathrm{bps/Hz}\right]\triangleq\mathbb{E}\left[\mathcal{C}\right]=\int_0^\infty\log_2\left(1+\eta\right)f_\eta(\eta)\mathrm{d}\eta, \label{eq-def-c}
\end{align}
where $\eta\in\left\{\gamma_{\mathrm{r}_i}^\mathrm{fama},\kappa_{\mathrm{r}_i}^\mathrm{fama}\right\}$ and $f_\eta\left(\eta\right)$ denotes the corresponding PDF. Therefore, the EC for the considered system model is derived and presented in the following theorem. 

\begin{theorem}
The EC for the FAMA-IC system under the strong interference scenario is given by
\begin{align}
\overline{\mathcal{C}}_1&\approx\log_2\left(1+\overline{\gamma}_\mathrm{r_1}H_1\left(1-\frac{\varpi^{\mathrm{r_1}}H_1}{2N_\mathrm{r_1}}\right)\right),\label{eq-c1}\\
\overline{\mathcal{C}}_2&\approx\log_2\left(1+\overline{\gamma}_\mathrm{r_2}H_2\left(1-\frac{\varpi^{\mathrm{r_2}}H_2}{2N_\mathrm{r_2}}\right)\right),\label{eq-c2}
\end{align}
and
\begin{multline}
\overline{\mathcal{C}}_\mathrm{sum} \approx \min\left\{\log_2\left(1+\overline{\kappa}_\mathrm{r_1}H_1\left(1-\frac{\varpi^{\mathrm{r_1}}H_1}{2N_\mathrm{r_1}}\right)\right)\right.,\\
\left.\log_2\left(1+\overline{\kappa}_\mathrm{r_2}H_2\left(1-\frac{\varpi^{\mathrm{r_2}}H_2}{2N_\mathrm{r_2}}\right)\right)\right\},\label{eq-csum}
\end{multline}
in which
\begin{align}
H_1=\sum_{n_\mathrm{r_1}=1}^{N_\mathrm{r_1}} \frac{1}{n_\mathrm{r_1}},~\text{and}~
H_2=\sum_{n_\mathrm{r_2}=1}^{N_\mathrm{r_2}} \frac{1}{n_\mathrm{r_2}},\label{eq-h}
\end{align}
where $\overline{\kappa}_{\mathrm{r}_i}=\mathbb{E}\left[\kappa_{\mathrm{r}_i}\right]=\overline{\gamma}_{\mathrm{r}_i}+\overline{\zeta}_{\mathrm{r}_i}$ defines the average of ${\kappa}_{\mathrm{r}_i}$.
\end{theorem}

\begin{proof}
In order to derive the exact EC, we need to insert the PDFs from \eqref{eq-pdf-gaussian} and \eqref{eq-pdf-gaussian-kappa} into \eqref{eq-def-c} and evaluate the integral. However, due to the complexity of $f_{\gamma^\mathrm{fama}_{\mathrm{r}_i}}\left(r\right)$ and  $f_{\kappa^\mathrm{fama}_{\mathrm{r}_i}}\left(r\right)$ that include the joint PDF of multivariate normal distribution, it is mathematically intractable to solve the provided integrals in closed form. For this purpose, we apply Jensen's inequality into \eqref{eq-pdf-gaussian} to derive a tight approximation of the EC, i.e.,
\begin{align}
\overline{\mathcal{C}}\approx\log_2\left(1+\mathbb{E}\left[\mathrm{\eta}\right]\right).\label{eq-jensen}
\end{align}
Therefore, we need to determine the expectation of $\gamma_{{\mathrm{r}}_i}^\mathrm{fama}$ and $\kappa_{{\mathrm{r}}_i}^\mathrm{fama}$, which involves the expectation of the maximum of $N_{\mathrm{r}_i}$ correlated RVs. To do so, we introduce a heuristic approach such that we first derive the expectation of the maximum of $N_{\mathrm{r}_i}$ independent RVs and then it is extended to the correlated scenario by adding a heuristic term based on the Gaussian copula \cite{falk2010laws,hennig2009expectation}. As for $\gamma_{{\mathrm{r}}_i}^\mathrm{ind}$, which includes $N_{\mathrm{r}_i}$ exponential RVs, the expected value for the independent case is given by
\begin{align}
\mathbb{E}\left[\gamma_{{\mathrm{r}}_i}^\mathrm{ind}\right]=\overline{\gamma}_{\mathrm{r}_i}H_i=\overline{\gamma}_{\mathrm{r}_i}\sum_{n_{\mathrm{r}_i}=1}^{N_{\mathrm{r}_i}}\frac{1}{n_{\mathrm{r}_i}},
\end{align}
where $H_i$ is the $N_{\mathrm{r}_i}$-th harmonic number. Now, when the exponential RVs become correlated, the expected maximum decreases because the dependence among the RVs reduces the variability of the maximum value. Therefore, this reduction in the expected maximum needs to be accounted for the correlation case. In our correlation model, for the Gaussian copula with the correlation matrix $\mathbf{R}_{\mathrm{r}_i}$, the parameter $\varpi^{\mathrm{r}_i}$ influences how the RVs are related. Therefore, given that the average correlation, e.g., $\varpi^{\mathrm{r}_i}$, among the RVs has a diminishing effect on the maximum values, a heuristic approximation of the expectation  $\gamma_{{\mathrm{r}}_i}^\mathrm{fama}$ for the correlated case can be given by
\begin{align}
\mathbb{E}\left[\gamma_{\mathrm{r}_i}^\mathrm{fama}\right]\approx\overline{\gamma}_{\mathrm{r}_i}H_i\left(1-\frac{\varpi^{\mathrm{r}_i}H_{i}}{2N_{\mathrm{r}_i}}\right). \label{eq-exp1}
\end{align}
The term $\left(1-\frac{\varpi^{\mathrm{r}_i}H_{i}}{2N_{\mathrm{r}_i}}\right)$ is based on the statistical principle of the Gaussian copula and heuristic arguments that adjust the expected value of the maximum to account for the correlation among the RVs. In other words, the term $\frac{\varpi^{\mathrm{r}_i}H_{i}}{2N_{\mathrm{r}_i}}$ is an approximation that adjusts the expectation for the reduction caused by the average pairwise correlation $\varpi^{\mathrm{r}_i}$ among the RVs.

Following the same approach, we first derive the expected value of $\kappa_{\mathrm{r}_i}^\mathrm{fama}$ that involves $N_{\mathrm{r}_i}$ independent hypoexponential RVs as
\begin{align}
\mathbb{E}\left[\kappa_{{\mathrm{r}}_i}^\mathrm{ind}\right]=\overline{\kappa}_{\mathrm{r}_i}H_i=\overline{\kappa}_{\mathrm{r}_i}\sum_{n_{\mathrm{r}_i}=1}^{N_{\mathrm{r}_i}}\frac{1}{n_{\mathrm{r}_i}},
\end{align}
where $\overline{\kappa}_{\mathrm{r}_i}=\overline{\gamma}_{\mathrm{r}_i}+\overline{\zeta}_{\mathrm{r}_i}$ is the average of $\kappa_{\mathrm{r}_i}$. Next, we adjust the harmonic number $H_i$ by the term $\frac{\varpi^{\mathrm{r}_i}H_{i}}{2N_{\mathrm{r}_i}}$ to account for the correlation among the correlated RVs. Thus, we have
\begin{align}
\mathbb{E}\left[\kappa_{\mathrm{r}_i}^\mathrm{fama}\right]\approx\overline{\kappa}_{\mathrm{r}_i}H_i\left(1-\frac{\varpi^{\mathrm{r}_i}H_{i}}{2N_{\mathrm{r}_i}}\right). \label{eq-exp2}
\end{align}
By substituting the derived expectations from \eqref{eq-exp1} and \eqref{eq-exp2} to \eqref{eq-jensen}, we have \eqref{eq-c1}--\eqref{eq-csum} and the proof is completed. 
\end{proof}

\begin{remark}It is noteworthy that the approximation provided for the maximum of correlated exponential and hypoexponential RVs using the Gaussian copula is a heuristic approach, which attempts to incorporate the effect of correlation in a manner analogous to methods used for normal distributions. Nevertheless, the accuracy of the proposed analytical approach needs to be evaluated. In this regard, the best way to validate this accuracy is through empirical simulations. By generating correlated exponential and hypoexponential RVs using the Gaussian copula and comparing the simulated expectations of the maximum to the approximated values, we can gauge the accuracy of the analytical derivations. 
\end{remark}

\begin{remark}
In Fig.~\ref{fig-ex}, we compare the approximation and the empirical results of the expected value for $\gamma_{{\mathrm{r}}_i}^\mathrm{fama}$ and $\kappa_{{\mathrm{r}}_i}^\mathrm{fama}$, respectively. It can be seen that the proposed approximations are perfectly matched with the empirical expectations. Furthermore, we can observe that the expected value of the maximum value of $N_{\mathrm{r}_i}$ exponential and hypoexponential RVs, i.e., $\mathbb{E}\left[\gamma_{{\mathrm{r}}_i}^\mathrm{fama}\right]$ and $\mathbb{E}\left[\kappa_{{\mathrm{r}}_i}^\mathrm{fama}\right]$, respectively, increases with $N_{\mathrm{r}_i}$ when considering the harmonic number $H_i=\sum_{n_{\mathrm{r}_i}=1}^{N_{\mathrm{r}_i}} \frac{1}{n_{\mathrm{r}_i}}$. This is because as $N_{\mathrm{r}_i}$ increases, the harmonic number $H_i$ grows, leading to a higher expected maximum.  
\end{remark}

\begin{figure}
\centering
\hspace{0cm}\subfigure[Exponential distribution]{%
\includegraphics[width=0.5\textwidth]{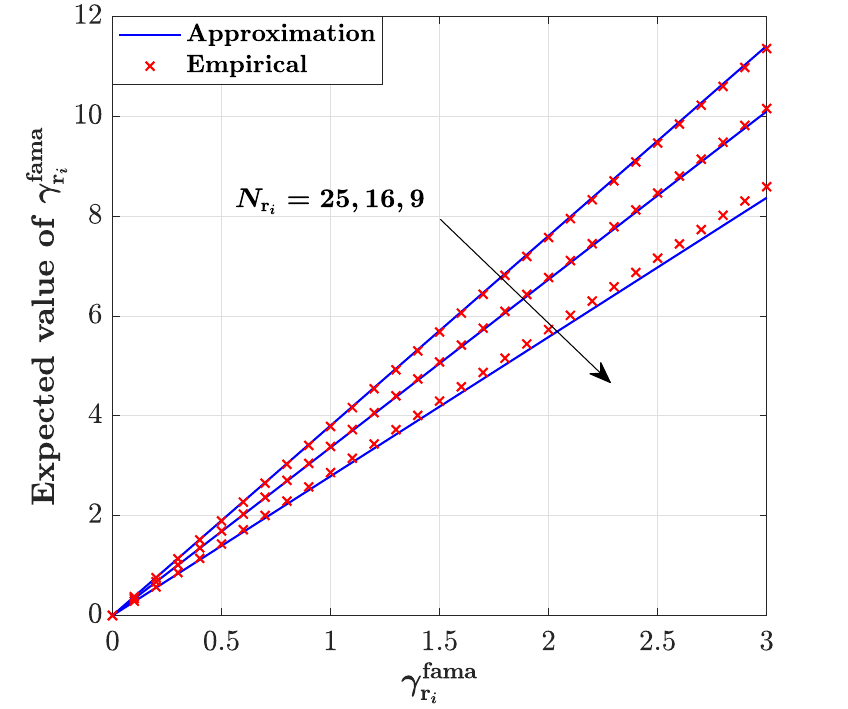}\label{fig_c_snr}%
}\hspace{0cm}
\subfigure[Hypoexponential distribution]{%
\includegraphics[width=0.5\textwidth]{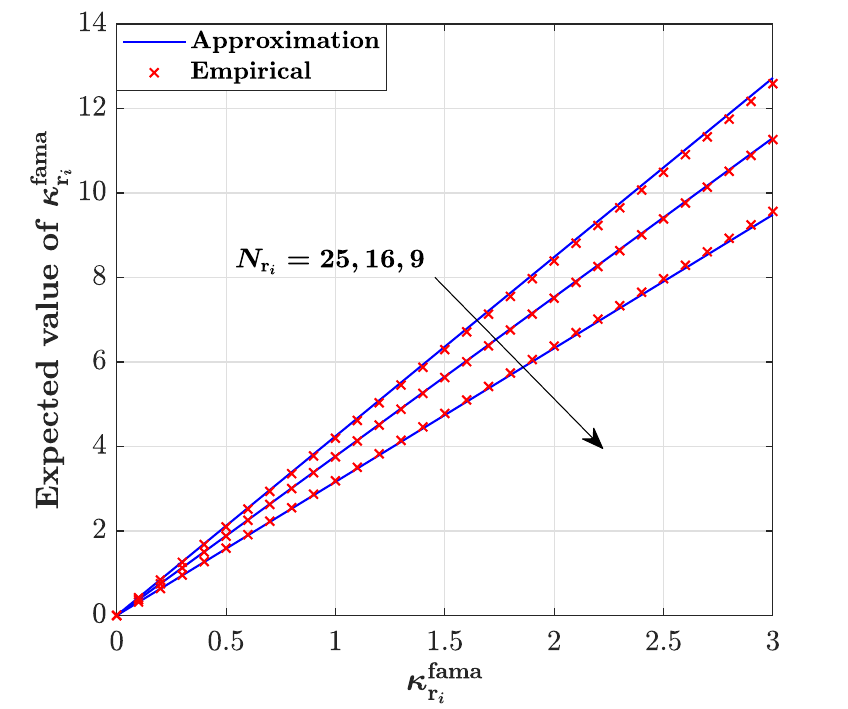}\label{fig_p_g}%
}\hspace{0cm}
\caption{The expected values of RVs (a) $\gamma_{{\mathrm{r}}_i}^\mathrm{fama}$ and (b) $\kappa_{{\mathrm{r}}_i}^\mathrm{fama}$.}\label{fig-ex}
\end{figure}

\begin{corollary}
The asymptotic EC in the high SNR regime for the FAMA-IC system under the strong interference scenario is given by
\begin{align}
\overline{\mathcal{C}}_1^\infty&\approx\log_2\left(\overline{\gamma}_\mathrm{r_1}H_1\left(1-\frac{\varpi^{\mathrm{r_1}}H_1}{2N_\mathrm{r_1}}\right)\right),\label{eq-c1asym}\\
\overline{\mathcal{C}}_2^\infty&\approx\log_2\left(\overline{\gamma}_\mathrm{r_2}H_2\left(1-\frac{\varpi^{\mathrm{r_2}}H_2}{2N_\mathrm{r_2}}\right)\right),\label{eq-c2-asym}
\end{align}
and
\begin{multline}
\overline{\mathcal{C}}_\mathrm{sum}^\infty \approx\min\left\{\log_2\left(\overline{\kappa}_\mathrm{r_1}H_1\left(1-\frac{\varpi^{\mathrm{r_1}}H_1}{2N_\mathrm{r_1}}\right)\right)\right.,\\
\left.\log_2\left(\overline{\kappa}_\mathrm{r_2}H_2\left(1-\frac{\varpi^{\mathrm{r_2}}H_2}{2N_\mathrm{r_2}}\right)\right)\right\},\label{eq-csum-asym}
\end{multline}
where $H_i$ has been defined in \eqref{eq-h}. 
\end{corollary}

\begin{proof}
In the high SNR regime, i.e., $\eta\gg 1$, in \eqref{eq-def-c}, we have $\log_2\left(1+\eta\right)\approx\log_2\left(\eta\right)$. Now, by directly applying this approximation into \eqref{eq-c1}--\eqref{eq-csum}, the proof is accomplished. 
\end{proof}

\begin{figure}[!t]
\centering
\includegraphics[width=0.95\columnwidth]{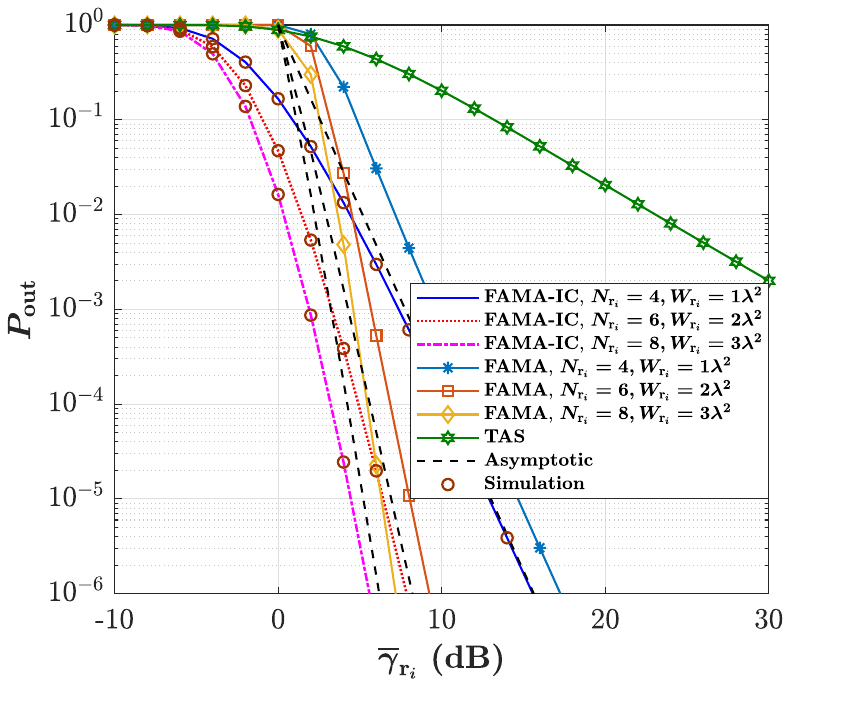}
\caption{OP versus average SNR $\overline{\gamma}_{\mathrm{r}_i}$ for different number of fluid antenna ports, $N_{\mathrm{r}_i}$, and fluid antenna size, $W_{\mathrm{r}_i}$, when $R_1^\mathrm{th}=R_2^\mathrm{th}=0.5$ bits, and  $\overline{\zeta}_{\mathrm{r}_1}=\overline{\zeta}_{\mathrm{r}_2}=\overline{\gamma}_{\mathrm{r}_{\iota}}+20$ dB.}\vspace{0cm}\label{fig-out_gamma}
\end{figure}

\begin{figure}[!t]
\centering
\includegraphics[width=0.95\columnwidth]{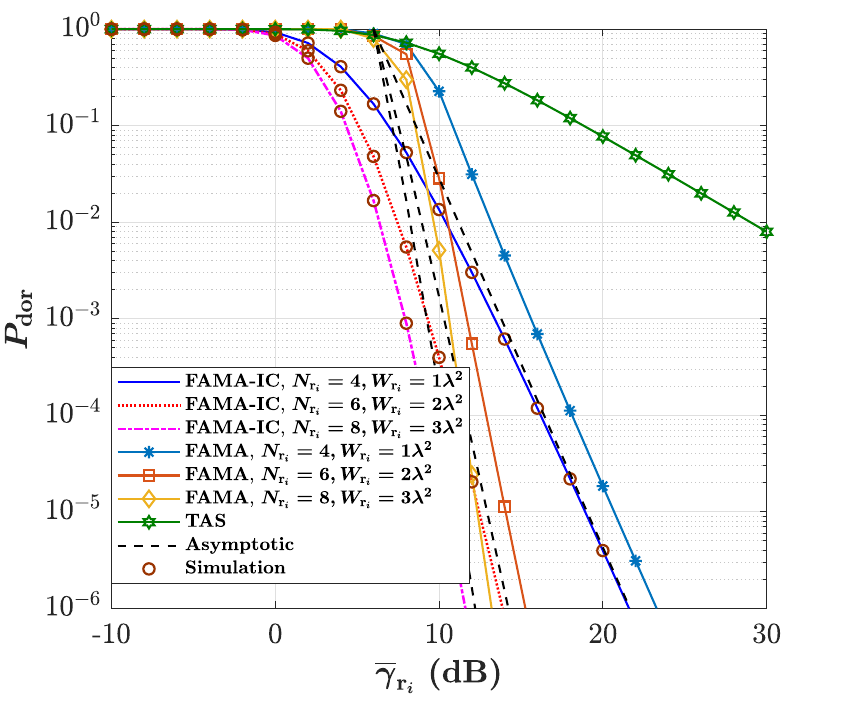}
\caption{DOR versus average SNR $\overline{\gamma}_{\mathrm{r}_i}$ for different number of fluid antenna ports, $N_{\mathrm{r}_i}$, and fluid antenna size, $W_{\mathrm{r}_i}$, when $R_1=R_2=1$ Kbits, $B_1=B_2=1$ MHz,  $T_1^\mathrm{th}=T_2^\mathrm{th}=1$ ms, and  $\overline{\zeta}_{\mathrm{r}_1}=\overline{\zeta}_{\mathrm{r}_2}=\overline{\gamma}_{\mathrm{r}_{\iota}}+20$ dB.}\vspace{0cm}\label{fig-dor_gamma}
\end{figure}

\section{Numerical Results}\label{sec-num}
In this section, we evaluate the analytical derivations, which are double-checked in all instances by Monte Carlo simulation. Figs.~\ref{fig-out_gamma} and \ref{fig-dor_gamma} show the performance of OP and DOR against the average SNR $\overline{\gamma}_{\mathrm{r}_i}$ for different values of fluid antenna size, $W_{\mathrm{r}_i}$, and number of fluid antenna ports, $N_{\mathrm{r}_i}$. First, we see that the asymptotic results closely match the numerical results at high SNRs. Moreover, assuming a strong interference scenario, i.e., $\overline{\zeta}_{\mathrm{r}_i}>\overline{\gamma}_{\mathrm{r}_{\iota}}$, as $\gamma_{\mathrm{r}_i}$ grows, the OP and DOR for both FAS and TAS decrease which is expected since the channel condition and the quality of the desired signal improve. It is also evident that the performance of OP and DOR enhances as $W_{\mathrm{r}_i}$ and $N_{\mathrm{r}_i}$ increase. This improvement is mainly because increasing the fluid antenna size increases the spatial separation between the ports, and therefore, reduces the spatial correlation. Although increasing the number of ports leads to stronger spatial correlation, it can potentially enhance the channel capacity, diversity gain, and spatial multiplexing. In other words, increasing $N_{\mathrm{r}_i}$ while keeping $W_{\mathrm{r}_i}$ fixed causes the ports to become too close to one another, which leads to dominating spatial correlation effects. As a result, the diversity gain diminishes beyond a certain point, and the reduction in OP and DOR slows down, eventually reaching saturation. Thus, by simultaneously increasing $W_{\mathrm{r}_i}$ and $N_{\mathrm{r}_i}$, the spatial correlation between the fluid antenna ports becomes balanced; consequently, a lower OP and DOR is provided. Furthermore, it can be seen that deploying a fluid antenna with only one activated port at users obtains significantly lower OP and DOR compared to TAS in the considered FAMA-IC with strong interference scenario. For instance, for a fixed average SNR $\overline{\gamma}_{\mathrm{r}_1}=10$ dB, the OP for user $\mathrm{r}_i$ under FAMA is in the order of $10^{-4}$, while it is around $10^{-1}$ under TAS. Moreover, by setting $R_i=1$ Kbits, $B_i=1$ MHz, $T_i^\mathrm{th}=1$ ms, and $\overline{\gamma}_{\mathrm{r}_i}$, we can observe that the DOR for user $r_i$ under FAMA is around $10^{-2}$, whereas it is near to $1$ for the TAS counterpart. 

Additionally, a comparison of the results for FAMA-IC and the benchmark FAMA reveals that FAMA-IC outperforms FAMA in terms of both OP and DOR under conditions of strong interference because it leverages fluid antenna diversity to optimize the SINR while employing advanced decoding techniques, such as SND, to mitigate interference. This approach enables the system to maintain reliable data rates by utilizing the structure of the strong interference instead of being fully degraded by it, which would occur if interference were treated as noise. As a result, the increased channel capacity achievable under strong interference reduces the likelihood of falling below the threshold rate, making strong interference advantageous for minimizing OP in FAMA systems.

\begin{figure}[!t]
\centering
\includegraphics[width=0.95\columnwidth]{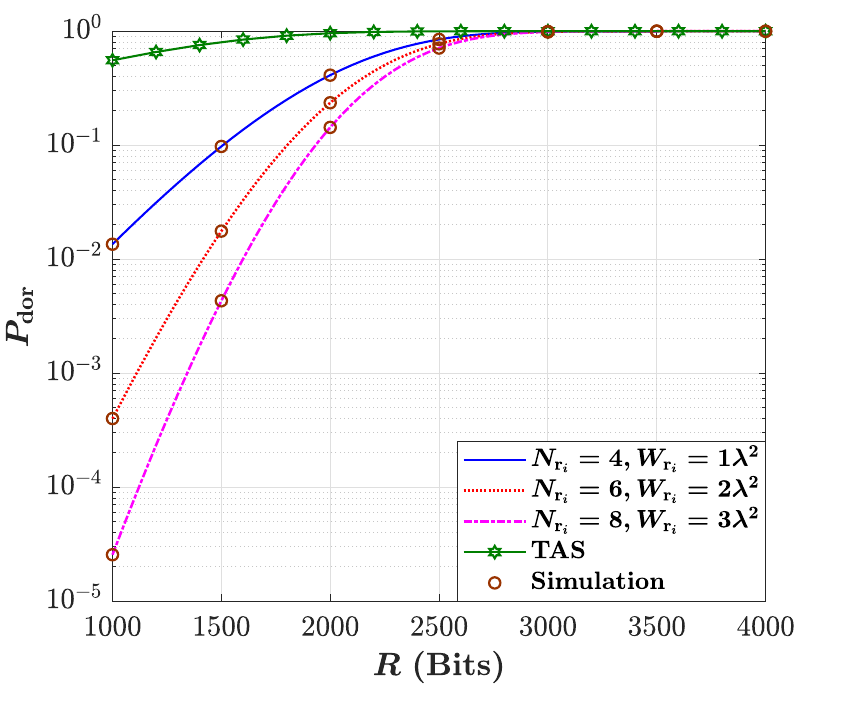}
\caption{DOR versus rate $R$ for different number of fluid antenna ports, $N_{\mathrm{r}_i}$, and fluid antenna size, $W_{\mathrm{r}_i}$, when $B_1=B_2=1$ MHz, $\overline{\gamma}_{\mathrm{r}_i}=10$ dB, and $T_1^\mathrm{th}=T_2^\mathrm{th}=1$ ms.}\vspace{0cm}\label{fig-dor_rate}
\end{figure}

\begin{figure}[!t]
\centering
\includegraphics[width=0.95\columnwidth]{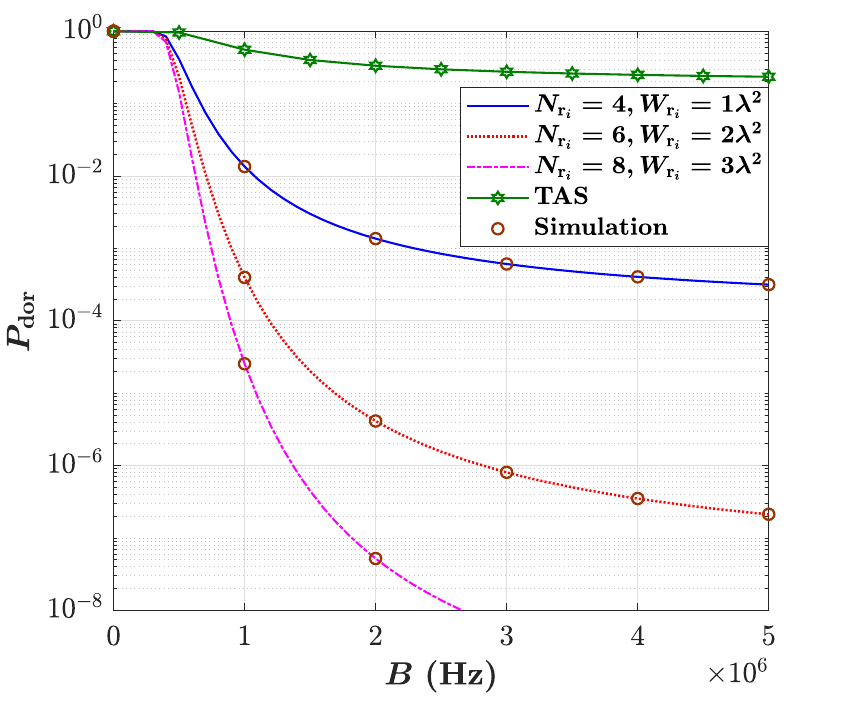}
\caption{DOR versus bandwidth $B$ for different number of fluid antenna ports, $N_{\mathrm{r}_i}$, and fluid antenna size, $W_{\mathrm{r}_i}$, when $R_1=R_2=1$ Kbits, $\overline{\gamma}_{\mathrm{r}_i}=10$ dB, and $T_1^\mathrm{th}=T_2^\mathrm{th}=1$ ms.}\vspace{0cm}\label{fig-dor_band}
\end{figure}\vspace{0cm}

\begin{figure}[!t]
\centering
\includegraphics[width=0.95\columnwidth]{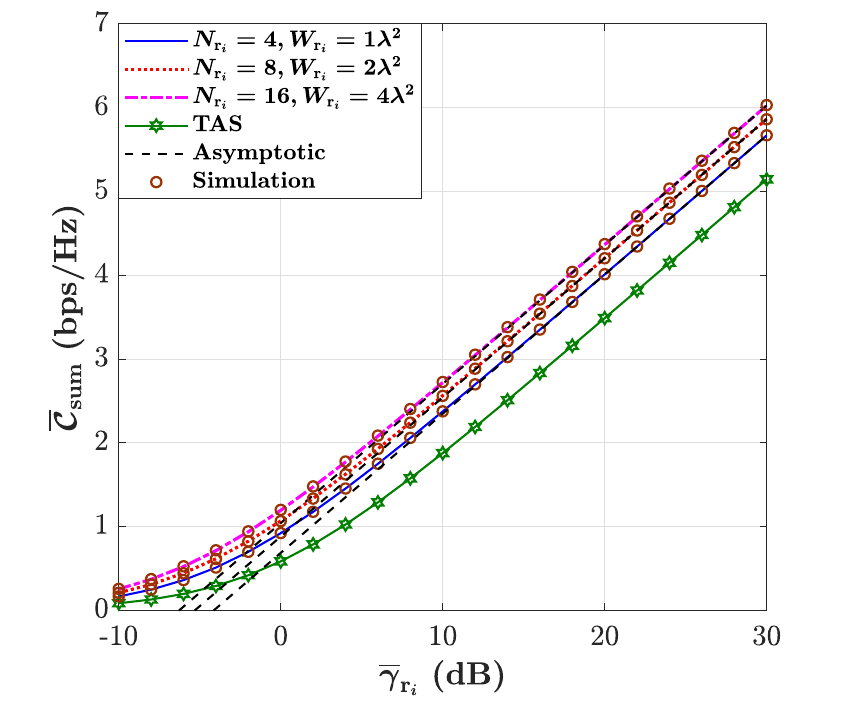}
\caption{EC versus average SNR $\overline{\gamma}_{\mathrm{r}_i}$ for different number of fluid antenna ports, $N_{\mathrm{r}_i}$, and fluid antenna size, $W_{\mathrm{r}_i}$, when $R_1^\mathrm{th}=R_2^\mathrm{th}=0.5$ bits, $\overline{\zeta}_{\mathrm{r}_1}=\overline{\zeta}_{\mathrm{r}_2}=\overline{\gamma}_{\mathrm{r}_{\iota}}+20$ dB.}\vspace{0cm}\label{fig-csum_gamma}
\end{figure}

Assuming $R_i=R$ and $B_i=B$, Fig.~\ref{fig-dor_rate} investigates the impact of increasing the data rate $R$ on the DOR performance under FAMA and TAS deployments in the strong IC. While transmitting a larger amount of data within a fixed bandwidth $B$ is challenging for both FAMA and TAS, it becomes much more feasible when using FAS. For example, sending $R=1.5$ Kbits over a strong IC with $B=1$ MHz is nearly impossible with TAS, but it can be achieved with an acceptable DOR under FAMA. The impact of bandwidth $B$ on the DOR for selected values of $W_{\mathrm{r}_i}$ and $N_{\mathrm{r}_i}$ is presented in Fig.~\ref{fig-dor_band}. It can be seen that increasing $B$ improves the DOR performance for both FAMA and TAS in strong IC, which means that the same amount of data can be transmitted with less delay as the IC bandwidth increases. However, it can be observed that this improvement is much more remarkable under the FAMA deployment than the TAS, thanks to the position reconfigurability, enhanced spatial diversity, and precise interference management provided by FAS. For instance, with a predefined bandwidth of $B=1$ MHz, sending $1$ Kbits of data results in a DOR close to $10^{-5}$ for FAMA and nearly  $1$ for TAS. 

Fig.~\ref{fig-csum_gamma} studies the performance of the sum EC as a function of $\overline{\gamma}_{\mathrm{r}_i}$ for selected values of $W_{\mathrm{r}_i}$ and $N_{\mathrm{r}_i}$. It is evident that the asymptotic results closely align with the analytical derivations at high SNR, confirming the reliability of the theoretical model to predict the system performance. As expected, increasing $W_{\mathrm{r}_i}$ and $N_{\mathrm{r}_i}$ leads to a significant improvement in the sum EC. Even under the challenging strong IC, FAMA outperforms TAS, with the latter showing lower capacity across all SNR values. This highlights the advantage of using FAS with only one activated port and even with smaller sizes, over a fixed single-antenna in interference-dominated environments. However, in strong IC, the EC is constrained by the interference power relative to the signal power. As $W_{\mathrm{r}_i}$ and $N_{\mathrm{r}_i}$ increase, the system can initially mitigate interference more effectively by leveraging spatial diversity. Nevertheless, beyond a certain point, the interference remains dominant, and the system capacity cannot improve further, leading to saturation.

\section{Conclusion}\label{sec-con}
This paper has provided a comprehensive study of FAMA for two-user strong IC, with each mobile user equipped with a FAS. By applying the realistic SND interference management technique, analytical expressions for the CDF and  PDF of the SNR and several key performance metrics, including the OP, DOR, and EC were derived. The results demonstrated significant advantages of FAS, showcasing its ability to improve system performance even with only one activated port under strong interference conditions. These findings contribute to more realistic interference management models in future wireless communication systems and offer valuable insights for optimizing FAMA in practical deployments.


\end{document}